\documentclass[runningheads]{llncs}
\usepackage{amsmath,amssymb,amsfonts}
\usepackage{mathtools}

\usepackage{amsthm}
\usepackage[usenames,dvipsnames]{color}
\usepackage{xcolor}
\usepackage{xspace}
\usepackage{microtype}
\usepackage{todonotes}
\usepackage{colonequals}
 \usepackage{booktabs}

\usepackage{amssymb, bm}

\usepackage{multirow}
\usepackage{multicol}

\usepackage{paralist}
\usepackage{pgfplots,pgfplotstable}
\usepackage{url}
\usepackage{appendix}
\usepackage{graphicx}
\usepackage{algorithm}
\usepackage{listings}
\usepackage{nicefrac}
\usepackage{tikz} 
\pgfplotsset{compat=1.8}
\usepackage{mdframed}
\usepgfplotslibrary{fillbetween}
\usepgfplotslibrary{groupplots}
\usepgfplotslibrary{dateplot}
\usetikzlibrary{positioning}

\usetikzlibrary{patterns}

\usetikzlibrary{arrows,decorations.pathmorphing,positioning,fit,trees,shapes,shadows,automata,calc} 

\tikzset{outline/.style args={#1}{%
  draw=#1,thick,fill=#1!50}}

\title{Scenario-Based Verification of Uncertain MDPs
\thanks{Supported by the grants ARL \# ACC-APG-RTP W911NF, NASA \# 80NSSC19K0209,  
	NSF \# 1646522, and NSF \# 1652113.}
}

\author{Murat Cubuktepe\inst{1}, Nils Jansen\inst{2}, Sebastian Junges\inst{3},\\ Joost-Pieter Katoen\inst{3}, Ufuk Topcu\inst{1}}
\institute{The University of Texas at Austin, Austin, USA \and Radboud University Nijmegen, Nijmegen, The Netherlands \and RWTH Aachen University, Aachen, Germany}

\authorrunning{M.\ Cubuktepe et al.}

\newcommand{\ie}{i.e.\@\xspace}

\newcommand{\eg}{e.g.\@\xspace}

\newcommand{\tool}[1]{\textrm{#1}\xspace}

\newcommand{\subjectto}{\textnormal{subject to}}

\newcommand{\minimize}{\textnormal{minimize}}
\usepackage{array,booktabs}
\newcolumntype{C}{>{$\displaystyle}c<{$}}
\usepackage[normalem]{ulem}
\useunder{\uline}{\ul}{}

\newcommand{\PP}{\mathbb{P}}

\newcommand{\RR}{\mathbb{R}}

\newtheorem{thm}{Theorem}

\newcommand{\dtmc}{\mathcal{D}}

\newcommand{\p}{\ensuremath{\mathbf{P}}}
\newcommand{\probdist}{\ensuremath{\mathbb{P}}}

\newcommand{\reachPropg}[2]{\ensuremath{\p_{\leq #1}(\finally #2)}}

\newcommand{\satprob}{F(\vmdp,\varphi)}
\newcommand{\satreachprob}{F(\vmdp,\varphi_r)}
\newcommand{\satreachprobmc}{F(\vmc,\varphi_r)}
\newcommand{\satcostprob}{F(\vmdp,\varphi_c)}
\newcommand{\satcostprobmc}{F(\vmdp,\varphi_c)}

\newcommand{\falseprob}{F(\vmdp,\neg\varphi)}

\newcommand{\reachPropgT}{\ensuremath{\reachPropg{\lambda}{T}}}

\newcommand{\reachPropSymbol}{\varphi_r}
\newcommand{\ereachPropSymbol}{\varphi_c}

\newcommand{\er}{\ensuremath{\mathrm{EC}}}

\newcommand{\expRewProp}[2]{\ensuremath{\er_{\leq #1}(\finally #2)}}

\newcommand{\finally}{\lozenge}

\newcommand{\R}{\mathbb{R}}
\newcommand{\Q}{\mathbb{Q}}

\newcommand{\Ireal}{[0,\, 1]\subseteq\mathbb{R}}  

\newcommand{\Distr}{\mathit{Distr}}

\newcommand{\distDom}{X}

\newcommand{\distFunc}{\mu}
\newcommand{\distDomElem}{x}

\newcommand{\Paramvar}{\ensuremath{{V}}\xspace}        

\newcommand{\sinit}{s_{\mathit{I}}} 
\newcommand{\mdp}{\mathcal{M}}
\newcommand{\pmdp}{\mdp}
\newcommand{\vmdp}{\mdp_{\probdist}}

\newcommand{\vmc}{\dtmc_{\probdist}}

\newcommand{\vdtmc}{\vmc}

\newcommand{\probmdp}{\mathcal{P}}

\newcommand{\paramspace}[1][]{\ensuremath{\mathcal{V}_{#1}}}
\newcommand{\costspace}[1][]{\ensuremath{\mathcal{W}_{#1}}}

\newcommand{\sched}{\ensuremath{\sigma}}
\newcommand{\Sched}{\ensuremath{\mathit{Str}}}

\newcommand{\Act}{\ensuremath{\mathit{Act}}}
\newcommand{\ActS}{\ensuremath{\mathit{ActS}}}

\newcommand{\act}{\ensuremath{\alpha}}

\DeclareMathAlphabet{\mathpzc}{OT1}{pzc}{m}{it}
\def\presuper#1#2%
  {\mathop{}%
   \mathopen{\vphantom{#2}}^{#1}%
   \kern-\scriptspace%
   #2}

\begin{document}
\maketitle
\begin{abstract}
We consider Markov decision processes (MDPs) in which the transition probabilities and rewards belong to an uncertainty set parametrized by a collection of random variables.
The probability distributions for these random parameters are unknown.
The problem is to compute the probability to satisfy a temporal logic specification within  any MDP that corresponds to a sample from these unknown distributions.
In general, this problem is undecidable, and we resort to techniques from so-called scenario optimization. 
Based on a finite number of samples of the uncertain parameters, each of which induces an MDP, the proposed method estimates the probability of satisfying the specification by solving a finite-dimensional convex optimization problem. 
The number of samples required to obtain a high confidence on this estimate is independent from the number of states and the number of random parameters. 
Experiments on a large set of benchmarks show that a few thousand samples suffice to obtain high-quality confidence bounds with a high probability.
\keywords{MDP, Uncertainty, Verification, Scenario optimization}
\end{abstract}

\section{Introduction}
\label{sec:introduction}


\paragraph{MDPs.} Markov decision processes (MDPs) model sequential decision-making problems in stochastic dynamic environments~\cite{puterman2014markov}. 
They are widely used in areas like planning~\cite{russell2016artificial}, reinforcement learning~\cite{sutton2018reinforcement}, formal verification~\cite{BK08}, and robotics~\cite{mcallister2012motion}.
Mature model checking tools like \tool{PRISM}~\cite{KNP11} and \tool{Storm}~\cite{DBLP:conf/cav/DehnertJK017} employ efficient algorithms to verify the correctness of MDPs against temporal logic specifications~\cite{Pnu77} provided all transition probabilities and cost functions are exactly known. 
In many applications, however, this assumption may be unrealistic, as certain system parameters are typically not exactly known and under control by external sources.

\paragraph{Uncertain MDPs.}
A common approach to deal with unknown system parameters is to let transition probabilities and cost functions of an MDP belong to uncertainty sets, resulting in so-called \emph{uncertain} MDPs~\cite{DBLP:conf/cdc/WolffTM12,nilim2005robust,wiesemann2013robust}, which generalize \emph{interval} MDPs~\cite{DBLP:conf/lata/DelahayeLLPW11,puggelli2013polynomial,givan2000bounded}. 
However, solution approaches, e.g., in~\cite{DBLP:conf/cdc/WolffTM12,nilim2005robust,wiesemann2013robust}, usually rely on the potentially limiting assumption that the uncertainty sets at different states of the MDP are independent from each other. 

Consider a simple motion planning scenario where an unmanned aerial vehicle (UAV) is tasked to transport a certain payload to a target location. 
The problem is to compute a policy for the UAV to successfully deliver the payload while taking into account the weather conditions. 
External factors like wind strength or direction may affect the movement of the UAV.
The assumption that such weather conditions are independent between the different possible states of UAV is unrealistic, and does not adequately model the scenario at hand. 

For settings in which the uncertainties at different states depend on each other, an option is to account for all possible--albeit infinitely many--values in the uncertainty sets. 
The policy synthesis problem can be formulated as a so-called semi-infinite convex optimization problem, which includes finitely many variables but infinitely many constraints~\cite{wiesemann2013robust}. 
This problem, however, is NP-hard~\cite{wiesemann2013robust,bertsimas2011theory}. 
Furthermore, it fails to exploit additional information that may be available as random variables over the uncertainty sets~\cite{DBLP:conf/valuetools/Scheftelowitsch17}, and may be very conservative. 
For instance, weather-data in the form of probability distributions may provide additional information on potential changes during the mission.

In this paper, we study a setting in which the fact the uncertain parameters are random variables and the dependencies between them are accounted for explicitly. 
Furthermore, each random parameter follows an unknown probability distribution from which we can sample the parameter values.
\begin{mdframed}[backgroundcolor=gray!30]
\emph{Problem statement.} Compute the probability with which there exists a policy such that a reachability or an expected-cost specification is satisfied for any randomly drawn parameter value.
\end{mdframed}%
We call this probability the \emph{satisfaction probability}.
The intuition is that the question of whether all (or some) parameter values satisfy a specification---as is often done in parameter synthesis~\cite{DBLP:journals/corr/abs-1903-07993}---is replaced by the question of \emph{how much we expect the (sampled) model to satisfy a specification}.
For example, a satisfaction probability of $80\%$ tells that, if we randomly sample the parameters, with a probability of $80\%$ there exists a policy for the resulting MDP that satisfies the specification.
Computing the satisfaction probability is in general undecidable, even for known probability distributions over the parameter values~\cite{arming2018parameter}.

%


\paragraph{Scenario-based verification.}
Therefore, we resort to \emph{sampling-based} algorithms that yield a confidence (probability) on the bounds of the satisfaction probability.
Referring back to the UAV example, we want to compute a confidence probability in the probability that there exists a policy for the UAV to successfully finish the mission.
As a first step, we take the aforementioned semi-infinite optimization problem that accounts for all possible parameter values as a basis.
Each concrete parameter value is referred to as a \emph{scenario} in the convex optimization literature ~\cite{DBLP:journals/tac/CalafioreC06}.
For specific problems where a distribution over individual scenarios is present, a technique called \emph{scenario-based optimization} provides guarantees on the satisfaction probability via efficient sampling techniques~\cite{DBLP:journals/tac/CalafioreC06,DBLP:journals/siamjo/CampiG08}.
The basic idea is to consider a finite set of samples from the distribution over the scenarios and restrict the semi-infinite problem to these samples.
The resulting convex optimization problem with finitely many constraints can be solved efficiently~\cite{boyd_convex_optimization}.

For our setting, we first sample a finite number of parameter instantiations each of which induces a concrete MDP.
We can solve the synthesis problem for this MDP efficiently using, \eg, a probabilistic model checker.
Based on the results, we compute a satisfaction probability and an estimate of its potential error.
For example, a $90\%$ estimate in a satisfaction probability of $80\%$, means that the error is at most $10\%$. 
We show that the error in the estimate diminishes to zero exponentially rapidly with increasing number of samples.
Moreover, we show that the number of required samples does neither depend on the size of the state space nor the number of random parameters.
We validate the theoretical results using several MDPs that have different sizes of state and parameter spaces and demonstrate experimentally that the required number of samples is indeed not sensitive to the dimension of the state and parameter space.
In addition, we show the effectiveness of our method with a new dedicated case study based on the aforementioned UAV example which incorporates 2\,500 random parameters.


\paragraph{Related work.}
The so-called~\emph{parameter synthesis} problem is concerned with computing parameter values such that there exists a policy in the induced non-parametric MDP that satisfies the specifications.
Most of the work in parameter synthesis focus on finding one parameter value that satisfies the specification.
The approaches involve computing a rational function of the reachability probabilities~\cite{Daw04,param_sttt,DBLP:journals/corr/abs-1804-01872}, utilizing convex optimization~\cite{DBLP:conf/tacas/Cubuktepe0JKPPT17,atvaqcqp}, and sampling-based methods~\cite{chen2013model,meedeniya2014evaluating}.
The problem of whether there exists a value in the parameter space that satisfies a reachability specification is ETR-complete\footnote{The ETR satisfiability problem is to decide if there exists a satisfying assignment to the real variables in a Boolean combination of a set of polynomial inequalities. 
It is known that NP $\subseteq$ ETR $\subseteq$ PSPACE.}~\cite{winkler2019complexity}, and finding a satisfying parameter value is exponential in the number of parameters.

The work in~\cite{bacci2019model} considers the analysis of Markov models in the presence of uncertain rewards, utilizing statistical methods to reason about the probability of a parametric MDP satisfying an expected cost specification.
This approach is restricted to reward parameters and does not explicitly compute confidence bounds.
\cite{DBLP:conf/sigsoft/LlerenaBBSR18}~computes bounds on the long-run probability of satisfying a specification with probabilistic uncertainty for Markov chains.
Other related techniques include multi-objective model checking to maximize the average performance with probabilistic uncertainty sets~\cite{DBLP:conf/valuetools/Scheftelowitsch17}, sampling-based methods which minimize the \emph{regret} with uncertainty sets~\cite{DBLP:journals/jair/AhmedVLAJ17}, and Bayesian reasoning to compute parameter values that satisfy a metric temporal logic specification on a continuous-time Markov chain~\cite{DBLP:conf/tacas/BortolussiS18}.
\cite{arming2018parameter} considers a variant of the problem in this paper where the probability distribution of the uncertainty sets is assumed to be known.
The paper formulates the policy synthesis problem as an (undecidable~\cite{chatterjee2016decidable}) partially observable Markov decision process (POMDP) synthesis problem and use off-the-shelf point-based POMDP methods~\cite{pineau2003point,cassandra1997incremental}.
%
%
The work in~\cite{puggelli2013polynomial,DBLP:conf/cdc/WolffTM12} consider the verification of MDPs with convex uncertainties. 
However, the uncertainty sets for different states in an MDP are restricted to be independent, which does not hold in our problem setting where we have parameter dependencies.

Uncertainties in MDPs have received quite some attention in the artificial intelligence and planning literatures. Interval MDPs~\cite{puggelli2013polynomial,givan2000bounded} use probability intervals in the transition probabilities. Dynamic programming, robust value iteration and robust policy iteration have been developed for MDPs with uncertain transition probabilities whose parameters are statistically independent, also referred to as rectangular, to find a policy ensuring the highest expected total reward at a given confidence level~\cite{nilim2005robust,DBLP:conf/cdc/WolffTM12}. The work in~\cite{wiesemann2013robust} relaxes this independence assumption a bit and determines a policy that satisfies a given performance with a pre-defined confidence provided an observation history of the MDP is given by using conic programming. State-of-the art exact methods can handle models of up to a few hundred of states~\cite{ho2018fast}. Multi-model MDPs~\cite{steimle2018multi} treat distributions over probability and cost parameters and aim at finding a single policy maximizing a weighted value function. For deterministic policies this problem is NP-hard, and it is PSPACE-hard for history-dependent policies.

\section{Preliminaries}


A \emph{probability distribution} over a finite set $\distDom$ is a function $\distFunc\colon\distDom\rightarrow\Ireal$ with $\sum_{\distDomElem\in\distDom}\distFunc(\distDomElem)=1$. 
The set of all distributions on $\distDom$ is denoted by $\Distr(\distDom)$. Let $\Paramvar= \lbrace x_1,\ldots,x_n\rbrace$ be a finite set of \emph{parameters} over $\mathbb{R}^n$. The set of polynomials over $\Paramvar$ is denoted by $\mathbb{Q}[\Paramvar].$ 
We denote the cardinality of a set $\mathcal{U}$ by $\vert \mathcal{U} \vert$. 

\subsection{Parametric Models}
\begin{definition}[pMDP]
	A \emph{parametric Markov decision process} (pMDP) $\pmdp$  is a tuple $\pmdp = (S, \Act, \sinit, \Paramvar, \mathcal{P})$
	with a finite set $S$ of \emph{states}, a finite set $\Act$ of \emph{actions}, an \emph{initial state} $\sinit \in S$, a finite set $\Paramvar$ of real-valued variables \emph{(parameters)} and a \emph{transition function} $\probmdp \colon S \times \Act \times S \rightarrow \mathbb{Q}[\Paramvar]$.
\end{definition}


%
%
%
For $s \in S$,  $\ActS(s) = \{\act \in \Act \mid \exists s'\in S,\,\probmdp(s,\,\act,\,s') \neq 0\}$ is the set of \emph{enabled} actions at $s$.
Without loss of generality, we require $\ActS(s) \neq \emptyset$ for $s\in S$.
If $|\ActS(s)| = 1$ for all $s \in S$, $\mdp$ is a \emph{parametric discrete-time Markov chain (pMC)}. 
We denote the transition function for pMCs by $\probmdp(s,s')$.

A pMDP $\mdp$ is a \emph{Markov decision process (MDP)} if the transition function yields \emph{well-defined} probability distributions, \ie, $\probmdp \colon S \times \Act \times S \rightarrow [0,1]$ and $\sum_{s'\in S}\probmdp(s,\act,s') = 1$ for all $s \in S$ and $\act \in \ActS(s)$.
We denote the \emph{parameter space of $\mdp$} by $\paramspace[\mdp]$.
Applying an \emph{instantiation} $u \in \paramspace[\pmdp]$ to a pMDP $\mdp$ yields the \emph{instantiated} MDP $\mdp[u]$ by replacing each $f\in\mathbb{Q}[\Paramvar]$ in $\mdp$ by $f[u]$. 
An instantiation $u$ is \emph{well-defined} for $\mdp$ if the resulting model $\mdp[u]$ is an MDP.
We assume that all parameter instantiations in $\paramspace[\mdp]$ yield well-defined MDPs.
We call $u$ \emph{graph-preserving} if for all $s,s'\in S$ and $\act\in \Act$ it holds that $\probmdp(s,\act,s')\neq 0 \Rightarrow \probmdp(s,\act,s')[u]\in(0,1]$.
%
If $\probmdp(s,\act,s') \in \{ p, 1-p \mid p \in \Paramvar \} \cup \Q$, then the parameter space $\paramspace[\mdp]$ is given by the rectangle $[0,1]^{|\Paramvar|}$.
We also consider a state-action cost function $c \colon  S \times Act \rightarrow \mathbb{Q}[\Paramvar]$. 
We denote the set of cost parameters as $\costspace$.

To define measures on
MDPs, nondeterministic choices are resolved by a so-called \emph{policy} $\sched\colon S\rightarrow\Act$ with $\sched(s) \in \ActS(s)$.
The set of all policies over $\mdp$ is $\Sched^\mdp$.
%
%
%
For the specifications that we consider in this paper, memoryless deterministic policies are sufficient~\cite{BK08}.
Applying a policy to an MDP yields an \emph{induced Markov chain} where all nondeterminism is resolved.
%

For an MC $\dtmc$, the \emph{reachability specification} $\reachPropSymbol=\reachPropgT$ asserts that a set $T \subseteq S$ of \emph{target states} is reached with probability at most $\lambda\in [0,1]$\footnote{The theory also applies to lower bounded properties.}.
If $\reachPropSymbol$ holds for $\dtmc$, we write $\dtmc\models\reachPropSymbol$.
Model checking for the more general PCTL~\cite{hansson_pctl} or $\omega$-regular specifications is often reducible to checking reachability specifications~\cite{BK08}. For an MDP $\mdp$, $\reachPropSymbol$ holds if for all $\sched\in \Sched^\mdp$ such that the induced MC $\dtmc$ by the policy $\sched$ reaches the set $T$ with a probability of at most $\lambda$. 
For an \emph{expected cost specification} $\ereachPropSymbol=\expRewProp{\kappa}{G}$, it holds that $\dtmc\models\ereachPropSymbol$ if and only if the expected cost of reaching a set $G \subseteq S$ is at most $\kappa \in \R$. 
The expected cost of reaching $G$ is well-defined if and only if $\mathbf{P}(\finally T)=1$ for all policies in an MDP.



\subsection{Uncertain MDPs}
We now introduce the setting that we study in this paper. 
Specifically, we use parameters to define the uncertainty in the transition probabilities and cost functions of an MDP. 
Each random parameter follows an unknown probability distribution from which we can sample the parameter values. 
%
\begin{definition}[uMDP]
	An \emph{uncertain Markov decision process} $\vmdp$ (uMDP) is a tuple $\vmdp = \left(\pmdp,   \probdist \right)$ where $\pmdp$ is a pMDP, and  $\probdist$ is a probability distribution over the parameter space $\paramspace[\pmdp]$.
	If $\pmdp$ is a pMC, then we call $\vmdp$ a uMC.
\end{definition}
Intuitively, a uMDP is a pMDP with an associated distribution over possible (graph-preserving) parameter instantiations.
That is, a realization of $\probdist$ yields a concrete MDP $\pmdp[u]$ with the respective instantiation $u\in\paramspace[\pmdp]$ (and $\probdist(u)>0$). 
\begin{remark}
In a uMDP, we distinguish \emph{controllable} and \emph{uncontrollable} parameters.
The uncontrollable parameters follow the probability distribution $\probdist$. 
In contrast, we can actively \emph{instantiate} the controllable parameters. In the following, we specifically allow cost parameters to be controllable. 
\end{remark}
\begin{definition}[Satisfaction Probability]
Let $\vmdp = \left(\pmdp, \probdist\right)$ be a uMDP and $\varphi$ a specification.
	The (weighted) \emph{satisfaction probability of $\varphi$} is  
	\[ \satprob= \int_{\paramspace[\pmdp]}   I_\varphi(u)\; d\,\probdist(u)\]
	with $u\in \paramspace[\pmdp]$ and $I_\varphi \colon \paramspace[\pmdp] \rightarrow \{0,1\}$ is the indicator for $\varphi$, i.e.\ $I_\varphi(u) = 1$ iff 
$\pmdp[u] \models \varphi$. 
\end{definition} 
Note that $I_\varphi$ is measurable, as $\paramspace[\pmdp]$ is the finite union of semi-algebraic sets~\cite{Book_Basu_RAAlgorithms}. 
Moreover, we have that $\satprob \in [0,1]$ and $\satprob + \falseprob= 1$.

\begin{example}
Consider the uMC in the left figure of Fig.~\ref{fig:exampleparam} with the uncontrollable parameter set $V=\{v\}$, initial state $s_0$, target set $T = \{s_3\}$ and an uniform distribution for the parameter $v$ over the interval $[0,1]$. We plot the probability of satisfying the specification $\reachPropSymbol=\reachPropgT$ as a function of $v$ in the right figure of Fig.~\ref{fig:exampleparam}. We also show the satisfying region and its complementary as green and red regions. The satisfying region is given by the union of the intervals $\left[0.13, 0.525\right]$ and $\left[0.89,1.0\right]$, and the satisfaction probability $\satreachprob$ is $0.395+0.11=0.505$.
\end{example} 

\begin{figure}[t]
\centering
\scalebox{0.9}{
\begin{tikzpicture}[scale=0.4, nodestyle/.style={draw,circle},baseline=(s6)]
    \node [nodestyle,initial] (s0) at (0,0) {$s_0$};
    \node [nodestyle] (s1) [on grid, right=1.8cm of s0] {$s_1$};
    \node [nodestyle] (s2) [on grid, right=1.8cm of s1] {$s_2$};
    \node [nodestyle ,accepting] (s3) [on grid, right=1.5cm of s2] {$s_3$};
    \node [nodestyle] (s4) [on grid, below=0.9cm of s1] {$s_4$};
        \node [nodestyle] (s5) [on grid, right=1.8cm of s4] {$s_5$};
    \node [nodestyle,gray] (s6) [on grid,below=1.0cm of s4,xshift=0.9cm] {$s_6$};
        \node [nodestyle,gray] (s7) [on grid,above=1.0cm of s1,xshift=0.9cm] {$s_7$};

    \draw[->] (s0) -- node [auto] {\scriptsize $1-v$} (s1);
    \draw[->,] (s0) -- node [below,pos=0.3, yshift=-0.1cm] {\scriptsize $v$} (s4);

    \draw[->] (s1) -- node [below,yshift=-3.5pt] {\scriptsize $0.1\cdot(1-v)$} (s2);
        \draw[->,gray] (s2) -- node [right,xshift=-2pt,yshift=3pt] {\scriptsize $1-v$} (s7);
        \draw[->,gray] (s1) -- node [left,xshift=2pt,yshift=4pt] {\scriptsize $v+0.9\cdot(1-v$)} (s7);
        \draw[->,gray] (s5) -- node [right,xshift=1pt,yshift=-3pt] {\scriptsize $v$} (s6);
        \draw[->,gray] (s4) -- node [left,xshift=-2pt,yshift=-3pt] {\scriptsize $1-0.5\cdot v^2$} (s6);

    \draw[->] (s2) -- node [auto] {\scriptsize $v$} (s3);
        \draw[->] (s4) -- node [auto,yshift=-1.25pt] {\scriptsize $0.5\cdot v^2$} (s5);
        \draw[->] (s5) -- node [right,xshift=1pt,yshift=-5pt] {\scriptsize $1-v$} (s3);

   \draw(s3) edge[loop right] node [right] {\scriptsize $1$} (s3);
      \draw(s6) edge[loop right,gray] node [right] {\scriptsize $1$} (s6);
   \draw(s7) edge[loop right,gray] node [right] {\scriptsize $1$} (s7);
\end{tikzpicture}}%
\scalebox{0.9}{
\begin{tikzpicture}[baseline]
\definecolor{color1}{rgb}{0.1,0.498039215686275,0.9549019607843137}
\begin{axis}[width=2.2in,height=1.9in,ymin=-0.0001,xmin=-0.01,xmax=1.01,xmajorgrids,ymajorgrids,xlabel={$v$},ylabel={$\reachPropgT$},ytick={0,0.05,0.1,0.13,0.15},yticklabels={0,0.05,0.1,$\lambda$,0.15},clip=false]

\addplot [color=color1,line width=0.5pt,dashed,domain=0:1, samples=101,unbounded coords=jump, name path=C]{0.5*x^3*(1-x) + 0.1*x^6*(1-x) +  0.1*x*(1-x)^3 +  0.1*(1-x)^2 + 0.04 + 0.05*x};
\addplot [color=color1,line width=1.5pt,domain=0:0.15, samples=101,unbounded coords=jump, name path=CC]{0.5*x^3*(1-x) + 0.1*x^6*(1-x) +  0.1*x*(1-x)^3 +  0.1*(1-x)^2 + 0.04 + 0.05*x};
\addplot [color=gray!50!white,line width=0pt,domain=0:0.15, samples=10,unbounded coords=jump, name path=DD]{0};
\addplot [color=gray!50!white,line width=0pt,domain=0.15:0.525, samples=10,unbounded coords=jump, name path=EE]{0};
\addplot [color=color1,line width=1.5pt,domain=0.525:0.89, samples=101,unbounded coords=jump, name path=CCC]{0.5*x^3*(1-x) + 0.1*x^6*(1-x) +  0.1*x*(1-x)^3 +  0.1*(1-x)^2 + 0.04 + 0.05*x};
\addplot [color=gray!50!white,line width=0pt,domain=0.525:0.89, samples=10,unbounded coords=jump, name path=DDD]{0};
\addplot [color=gray!50!white,line width=0pt,domain=0.89:1, samples=10,unbounded coords=jump, name path=EEE]{0};
\addplot [color=gray!50!white,line width=0pt,domain=0:0.15, samples=10,unbounded coords=jump, name path=DDA]{0.13};
\addplot [color=gray!50!white,line width=0pt,domain=0.15:0.525, samples=10,unbounded coords=jump, name path=EEA]{0.13};
\addplot [color=gray!50!white,line width=0pt,domain=0.525:0.89, samples=10,unbounded coords=jump, name path=DDDA]{0.13};
\addplot [color=gray!50!white,line width=0pt,domain=0.89:1, samples=10,unbounded coords=jump, name path=EEEA]{0.13};
     \addplot[red!20!white,domain=0:0.15] fill between[of=DD and DDA];
    \addplot[red!20!white,domain=0.525:0.89]  fill between[of=DDD and DDDA];
         \addplot[green!20!white,domain=0.15:0.525] fill between[of=EE and EEA];
    \addplot[green!20!white,domain=0.89:1]  fill between[of=EEE and EEEA];
    \draw [color=black,dashed,line width=1.5pt](axis cs:0,0.0) -- node[left]{} (axis cs:0,0.14);
    \draw [color=black,dashed,line width=1.5pt](axis cs:0.15,0.0) -- node[left]{} (axis cs:0.15,0.13);
    \draw [color=black,dashed,line width=1.5pt](axis cs:0.525,0.0) -- node[left]{} (axis cs:0.525,0.13);
    \draw [color=black,dashed,line width=1.5pt](axis cs:0.89,0.0) -- node[left]{} (axis cs:0.89,0.13);
    \draw [color=black,dashed,line width=1.5pt](axis cs:0,0.13) -- node[left]{} (axis cs:1.0,0.13);
%


\end{axis}
\end{tikzpicture}}%
\caption{Left: A uMC with parameter $v$. Right: The probability of satisfying the reachability specification $\reachPropSymbol=\reachPropgT$ versus the value of the parameter $v$. Intervals that satisfy $\reachPropSymbol$ are green, intervals that violate $\reachPropSymbol$ are red.}
\label{fig:exampleparam}
\end{figure}
\section{Problem Statement}

In this section, we state the problem that we study in this paper.
We seek to compute the satisfaction probability of the parameter space for a reachability or an expected cost specification $\varphi$ on a uMDP.
Intuitively, we seek the probability that a randomly sampled instantiation from the parameter space induces an MDP which satisfies $\varphi$.
Formally: Given a uMDP $\vmdp=\left(\pmdp,  \probdist\right)$, and a specification $\varphi$, compute the satisfaction probability $\satprob$. 
%
%
%
 %
%
However, as mentioned, the problem is in general undecidable~\cite{arming2018parameter}.
Therefore, we consider an approximation of computing the satisfaction probability:
\begin{mdframed}[backgroundcolor=gray!30]
\begin{problem}
\label{prob:primal_with_confidence}
Given a uMDP $\vmdp=\left(\pmdp,  \probdist\right)$, a reachability specification $\varphi_r = \reachPropgT$, and a \emph{tolerance probability} $\nu$, compute a confidence probability $\alpha_{\nu}$ such that $\satreachprob \geq 1 - \nu$ holds with a probability of at least $1-\alpha_{\nu}$.
  \end{problem}
\end{mdframed}
%
%
We illustrate the problem statement with the following example.
 \begin{example}
 For the UAV motion planning example, consider the question \textquotedblleft What is the probability on a given day such that there exists a policy for the UAV to successfully finish the mission.\textquotedblright~   
 A possible result is, e.g.,  0.78 (confidence probability: 0.99) and 0.81 (confidence probability: 0.95).
 Then, with a confidence probability of 0.99, the actual satisfaction probability is indeed greater than 0.78, and with a (slightly lower) confidence probability of 0.95 it is greater than 0.81. 
Such a result shows that it is quite likely that the UAV will finish the mission successfully with a probability that is at least 81\%.
 \end{example}
 
\noindent Similar to Problem~1, we also consider expected cost specifications.
 \begin{mdframed}[backgroundcolor=gray!30]
 \begin{problem}
 \label{prob:cost_with_confidence}
Given a uMDP $\vmdp=\left(\pmdp,  \probdist\right)$, and an expected cost specification  $\ereachPropSymbol=\expRewProp{\kappa}{G}$, a tolerance probability $\nu$, and a confidence probability $\alpha_\nu$  determine if there exists an instantiation to the cost parameters such that $\satcostprob\geq 1 - \nu$ holds with a probability of at least $1-\alpha_{\nu}$.
 \end{problem}
 \end{mdframed}
\begin{remark}
The main difference between Problem~\ref{prob:primal_with_confidence} and Problem~\ref{prob:cost_with_confidence} is that we consider \emph{controllable} cost parameters. 
We seek to compute an instantiation to these parameters such that the satisfaction probability is greater than $1-\nu$ with high confidence.
\end{remark}
\section{Scenario-Based Verification}
\label{sec:robust}
In this section, we present our approach to solving Problem~\ref{prob:primal_with_confidence} and~\ref{prob:cost_with_confidence}, that is, to approximate the satisfaction probability with respect to a specification.
We first consider the robust policy synthesis problem that accounts \emph{for all possible values} in the uncertainty set, potentially leading to a very pessimistic result.
This problem can be formulated as a semi-infinite convex optimization problem, which is NP-hard~\cite{wiesemann2013robust}.
Here, we exploit the structure of this problem, which includes finitely many variables but infinitely many constraints. 
Our approach is based on \emph{scenario optimization}~\cite{DBLP:journals/tac/CalafioreC06,DBLP:journals/siamjo/CampiG08}: We sample a finite number of parameter values and restrict the semi-infinite problem to these samples.
The resulting \emph{finite-dimensional} convex optimization problem can be solved efficiently~\cite{boyd_convex_optimization}.  
Based on the solution of the optimization problem, we compute high confidence in the estimate of the satisfaction probability. 
The estimate also generalizes to the samples from the probability distribution that are not in the sample set.  


\begin{remark}
For ease of presentation, we focus on uncertain Markov chains (uMCs).
Our results and methods generalize to uncertain MDPs (uMDPs). 
\end{remark}
We first develop the main results for the simple setting where \emph{all sampled} instantiated MCs from the parameter space $\paramspace[\dtmc]$ satisfy the reachability specification $\varphi_r$. 
This assumption does not imply that \emph{all} instantiated MCs satisfy $\varphi_r$: The sample set does not contain an MC that violates $\varphi_r$ even though there exists such an MC in the parameter space. 
In Section~\ref{sec:scenario:generalized}, we drop this assumption and allow sampled points in $\paramspace[\dtmc]$ to violate $\varphi_r$.
This completes our treatment of Problem~1.
In Section~\ref{sec:scenario:costs}, we show how our results generalize to expected cost specifications $\varphi_c$, to solve Problem~2. 



\subsection{Restriction to Satisfying Samples}
\label{sec:scenario:closetoone}
In this section, we assume that all instantiated MCs satisfy $\varphi_r$.
We then generalize our method to any values of $\nu$.
We want to check if a uMC $\dtmc$ satisfies a reachability specification $\reachPropSymbol=\reachPropgT$ for all instantiations in the sample set $\mathcal{U}$. 
For each instantiation, we can formulate a linear program (LP) that is feasible if and only if $\varphi_r$ is satisfied~\cite{puterman2014markov}. 
For a subset $\mathcal{U} \subseteq \paramspace[\dtmc]$ of the parameter space $\paramspace[\dtmc]$ of the uMC $\dtmc$, we can then write the conjunction of these LPs. 
We assume that $\vert \mathcal{U}\vert$ is finite and sampled from the probability distribution $\PP$ over the parameter space $\paramspace[\dtmc]$.

For each instantiation $u \in \mathcal{U}$, we introduce a set of linear constraints that are parametrized by $u$\footnote{we assume that each sample has a unique index}.
We use the following variables.
For $s \in S$ and $u \in \mathcal{U}$, the variable $p^u_s \in [0,1]$ represents the probability of reaching the target set $T \subseteq S$ from state $s$. 
The variable $\tau$ represents an upper bound on the probability of satisfying $\varphi_r$ for all instantiations in $\mathcal{U}$. Note that $\tau$ is a variable in our formulation, whereas $\lambda$ is the threshold of the reachability specification, and thus constant. 
The set $\neg\exists\finally T$ represents the set of states which cannot reach the target set $T$. 
The probability of reaching $T$ from these states is zero, and the set $\neg\exists\finally T$ does not change for different graph-preserving instantiations~\cite{param_sttt}.  
The set $\neg\exists\finally T$ can be found in polynomial time in the size of a uMC by using standard graph-based search algorithms~\cite{BK08}. 
We solve the following LP $\mathcal{L}_r(\mathcal{U})$, which is parametrized by each instantiation $u$ in $\mathcal{U}$,

\begin{align}
	&\displaystyle \minimize   \;\;	\tau 	\label{eq:robustmc_spec_obj}\\
		&\subjectto \quad \forall  u \in \mathcal{U}, \nonumber\\
		& p^u_{\sinit} \leq \tau,\label{eq:robustmc_spec_cons1}\\
	& p^u_{\sinit} \leq \lambda,\label{eq:robustmc_spec_cons2}\\
	&	 \displaystyle p^u_s = 1,\quad\forall s \in T,\label{eq:robustmc_spec_cons3}\\
		& \displaystyle  p^u_s = 0,\quad\forall s \in \neg\exists\finally T,\label{eq:robustmc_spec_cons4}\\
&	 \displaystyle  p^u_s = \sum\nolimits_{s' \in S}\mathcal{P}(s,s')[u] \cdot p^u_{s'},\quad\forall s \in S \setminus \left( T \cup \neg\exists\finally T \right).\label{eq:robustmc_spec_cons5}
\end{align}
The objective~\eqref{eq:robustmc_spec_obj} minimizes the maximal probability that can be achieved by all MCs induced by $\mathcal{U}$. 
The constraint~\eqref{eq:robustmc_spec_cons1} represents an upper bound on the reachability probability for all instantiations. We minimize the upper bound to compute the maximal probability of satisfying $\varphi_r$ for all instantiated MCs.
The constraint~\eqref{eq:robustmc_spec_cons2} ensures that the probability of reaching $T$ from the initial state $\sinit$ is below the threshold $\lambda$. 
The constraint~\eqref{eq:robustmc_spec_cons3} sets the probability to reach a state in $T$ from $T$ to 1. 
The constraint~\eqref{eq:robustmc_spec_cons4} sets the reachability probabilities from the states in $ \neg\exists\finally T$ to zero.
The constraint~\eqref{eq:robustmc_spec_cons5} computes the probability of satisfying the specification for each non-target state $s \in S$ in the standard way.  

There are infinitely many constraints in the semi-infinite LP $\mathcal{L}_r(\paramspace[\dtmc])$ as the cardinality of $(\paramspace[\dtmc])$ is infinite and $\mathcal{L}_r(\paramspace[\dtmc])$ has infinitely many constraints in the form of~\eqref{eq:robustmc_spec_cons1}--\eqref{eq:robustmc_spec_cons5}. 
Our approach is based on \emph{scenario optimization}~\cite{calafiore2005uncertain,DBLP:journals/tac/CalafioreC06,DBLP:journals/siamjo/CampiG08}, where we instantiate the parameters $u \in \paramspace[\dtmc]$ by sampling the probability distribution $\PP$. 
Then, for a given violation probability $\nu \in (0, 1)$, we compute a solution that violates the constraints in the LP $\mathcal{L}_r(\paramspace[\dtmc])$ with a probability that is not larger than $\nu$. 
We first give some properties of the LP $\mathcal{L}_r(\mathcal{U})$. 

\begin{thm}
	Let uMC $\dtmc$ and the sample sets $\mathcal{U} \subseteq \paramspace[\dtmc]$ with $K = \vert \mathcal{U} \vert \geq  2$. 
	Assume for all $u \in \mathcal{U}$, $\dtmc[u] \models \reachPropSymbol$.
	\label{thm:probabilistic_bound_1}
  For a given \emph{tolerance probability} $\nu\in[0,1)$, let the associated \emph{confidence probability}
    \begin{equation}
    \label{eq:tolerance_level}
\alpha_\nu=\sum\nolimits_{i=0}^{1}\dbinom{K}{i}(1-\nu)^{K-i}\nu^i.
\end{equation}
	Then, with a probability of at least $1-\alpha_\nu$, we have 
	\begin{equation}
	\label{eq:probabilistic_bound_1} \satreachprobmc \geq 1-\nu.
	\end{equation}
\end{thm}
\begin{proof}

The key idea of the proof is to relate the finite LP $\mathcal{L}_r(\mathcal{U})$ induced by a sampled set $\mathcal{U}$ to the semi-infinite LP $\mathcal{L}_r(\paramspace[\dtmc])$. 
Then, we use the results given in~\cite[Theorem 1]{DBLP:journals/siamjo/CampiG08} to obtain the lower bound $1-\alpha_\nu$. 
%
Let the convex set $C^{\vdtmc}_{\mathcal{U}}(\lambda, \tau)$ be generated by the set $\mathcal{U}$ according to the probability distribution $\PP$ over $\paramspace[\dtmc]$ as
	\begin{equation*}
	\begin{array}{l}
	C^{\vdtmc}_{\mathcal{U}}(\lambda, \tau) =\lbrace(\lambda,\tau) \mid ~\forall u \in \mathcal{U}  ~\textup{satisfying}~\eqref{eq:robustmc_spec_cons1}-\eqref{eq:robustmc_spec_cons5}\rbrace.\tag{$\star$}
	 \end{array}
	\end{equation*} 
The convex set $C^{\vdtmc}_{\mathcal{U}}(\lambda, \tau)$ constitutes the set of feasible instantiations to the LP $\mathcal{L}_r(\mathcal{U})$ and is exactly in the form of Equation 5 in~\cite{DBLP:journals/siamjo/CampiG08}. Using $C^{\vdtmc}_{\mathcal{U}}(\lambda, \tau)$, we reformulate $\mathcal{L}_r(\mathcal{U})$ as the convex program
	\begin{equation}
	\label{eq:robustmc_scenario}
\begin{array}{ccll}
&\minimize & \displaystyle \tau \\
&\subjectto 
& \displaystyle (\lambda,\tau)\in\mathcal C^{\vdtmc}_{\mathcal{U}}(\lambda, \tau),
\end{array}
\end{equation}
\noindent where the last constraint denotes that for a given $(\lambda,\tau)$, the feasible set of $C^{\vdtmc}_{\mathcal{U}}(\lambda, \tau)$ is not empty, i.e., there exists a feasible solution pair $(\lambda,\tau)$ to the scenario problem $\mathcal{L}_r(\mathcal{U})$.
This convex program asserts that all MCs in $\mathcal{U}$ should induce a reachability probability that is less than $\tau$, satisfying the specification $\varphi_r$.
Moreover, the convex program constitutes a scenario approximation to the so-called \emph{chance-constrained problem}~\cite{charnes1959chance}.  
Such an optimization problem states that the probability of satisfying a (chance) constraint is above a certain threshold:
\begin{equation}
\label{eq:robustmc_true}
\begin{array}{ccll}
&\minimize & \displaystyle \tau \\
&\subjectto & \displaystyle (\lambda,\tau)\in\RR\times\RR,\\
&& \displaystyle \PP\left((\lambda,\tau)\in\mathcal C^{\vdtmc}_{\paramspace[\dtmc]}(\lambda,\tau)\right)\geq 1-\nu.
\end{array}
\end{equation}
The chance constraint in~\eqref{eq:robustmc_true} ensures that the probability that an instantiation---obtained via distribution $\PP$---satisfies the specification $\varphi_r$ is at least~$1-\nu$. 
Theorem 1 in~\cite{DBLP:journals/siamjo/CampiG08} shows that any feasible solution to the problem in~\eqref{eq:robustmc_scenario} is feasible to the problem in~\eqref{eq:robustmc_true} with a confidence probability of $1-\alpha_\nu$, which shows that the violation probability of the solution is at most $\nu$. 
In our case, the probability of violation is exactly the probability that the instantiated MCs do not satisfy the specification $\varphi_r$. 
Thus, the claim follows. 
	\end{proof}

\begin{remark}[Independence to model size]\label{remark:independent}
	The confidence probability in Theorem~\ref{thm:probabilistic_bound_1} is in fact independent from the number of states, transitions, or random parameters of the uMC. 
	From a practical perspective, the number of samples that are needed for a certain confidence does not depend on the model size. 	
\end{remark}
	Finally, Theorem~\ref{thm:probabilistic_bound_1} asserts that with a probability of at least $1-\alpha_\nu$, the next sampled point from $\paramspace[\dtmc]$ will satisfy the specification with a probability of at least $1-\nu$. Note that $\alpha_\nu$ is the tail probability of a binomial distribution. 
	It converges exponentially rapidly to $0$ in $\vert  \mathcal{U} \vert$~\cite{DBLP:journals/siamjo/CampiG08}. 

\subsection{Satisfaction Probability by Treating Violating Samples}\label{sec:frac}
\label{sec:scenario:generalized}

Theorem~\ref{thm:probabilistic_bound_1} assumes that all sampled points, that is, the induced MCs, satisfy the specification $\varphi_r$. 
This is a severe assumption in general.
To lift this assumption, we consider the \emph{discarding approach} from~\cite{campi2011sampling}.
Specifically, after sampling a set of instantiations $\mathcal{U}$ from $\paramspace[\dtmc]$ according to the probability distribution $\PP$, we remove the constraints for the MCs that violate the specification $\varphi_r$ from the LP. 
We construct the set $\mathcal{R}=\mathcal{U}\setminus\mathcal{Q}$, where $\mathcal{Q}$ denotes the set of samples that induce MCs violating the specification $\varphi_r$.
Therefore, the set $\mathcal{R}$ denotes the set of sampled MCs that satisfy the specification $\varphi_r$. 
We then solve the LP $\mathcal{L}_r(\mathcal{R})$
\begin{equation}
	\label{eq:robustmc_spec_sample_discard}
	\begin{array}{llllllll}
	&\displaystyle \minimize \;\;  \tau \\
	&\subjectto \quad \forall u \in \mathcal{R},\\
	&\eqref{eq:robustmc_spec_cons1}-\eqref{eq:robustmc_spec_cons5},
	\end{array}
\end{equation}
where for $u \in \mathcal{R}$ and $s \in S$, $p^u_s$ gives the probability of satisfying the reachability specification of the instantiated MC $\dtmc[u]$ at state $s$.
The other constraints in the optimization problem in LP $\mathcal{L}_r(\mathcal{R})$ are identical to the LP $\mathcal{L}_r(\mathcal{U})$.
We give the main result of this section.

\begin{thm}
Let uMC $\dtmc$ and the sample sets $\mathcal{U},\mathcal{Q}\subseteq \paramspace[\dtmc]$,  with $K = \vert \mathcal{U} \vert \geq  2$ and $L = \vert \mathcal{Q}\vert $. 
	\label{thm:probabilistic_bound_discard}
  For a given \emph{tolerance probability} $\nu\in[0,1)$, the associated \emph{confidence probability} is
    \begin{equation}
    \label{eq:tolerance_level_disard}
\alpha_\nu=\dbinom{L+1}{L}\sum\nolimits_{i=0}^{L+1}\dbinom{K}{i}(1-\nu)^{K-i}\nu^i.
\end{equation}
	Then, with a probability of at least $1-\alpha_\nu$, we have	\begin{equation}
	\label{eq:probabilistic_bound_discard} \satreachprobmc \geq 1-\nu.
	\end{equation}
\end{thm}
\begin{proof}
Similar to the proof of Theorem 1, the main idea is to relate the LP $\mathcal{L}_r(\mathcal{R})$ to the chance-constrained convex problem in~\eqref{eq:robustmc_true}. Then, we invoke the results from~\cite[Theorem 1]{campi2011sampling} to get the desired result.
%
Let the convex set $C^{\vdtmc}_{\mathcal{R}}(\lambda, \tau)$, which is generated by the samples in $\mathcal{R}$, be defined by
	\begin{equation*}
	\begin{array}{l}
	C^{\vdtmc}_{\mathcal{R}}(\lambda, \tau)=\lbrace(\lambda,\tau) \mid	\forall u \in \mathcal{R}  \textup{ such that } (*) \textup{ is satisfied}\rbrace.
	 \end{array}
	\end{equation*}
	The set $C^{\vdtmc}_{\mathcal{R}}(\lambda, \tau)$ is in the form of the Definition 2.1 in~\cite{DBLP:journals/siamjo/CampiG08}.
We reformulate the LP $\mathcal{L}_r(\mathcal{R})$ as the convex program
	\begin{equation}
\begin{array}{ccll}
&\minimize & \displaystyle \tau \\
&\subjectto 
& \displaystyle (\lambda,\tau)\in C^{\vdtmc}_{\mathcal{R}}(\lambda, \tau).
\end{array}
\label{eq:scenario_arbitrary}
\end{equation}
where the last constraint denotes that the instantiated MCs from the parameter values of the set $\mathcal{R}$ should induce a reachability probability less than $\tau$, and thus, satisfy the specification $\varphi_r$. 
The problem in~\eqref{eq:scenario_arbitrary} is a scenario approximation to the problem in~\eqref{eq:robustmc_true}.
Theorem 2.1 in \cite{campi2011sampling} asserts that with a probability of $\alpha_\nu$, the violation probability of the solution is at most $\nu$, which is the probability of violating the specification for the next sample. 
Similar to Theorem~1, the violation probability $\nu$ is the probability that an instantiated MC does not satisfy the specification $\varphi_r$. Thus, the claim follows.
\end{proof}

\subsection{Expected Cost Specifications}
\label{sec:scenario:costs}
So far, we have focused on parameters that were uncontrollable, and assumed to be random.
Now, we consider the case where the cost function $c$ is parametric and the cost parameters are \emph{controllable}. 
Therefore, the parameters in the cost function are now variables that we can optimize over to satisfy an expected cost specification $\varphi_c=\expRewProp{\kappa}{G}$ for the instantiated MCs. 
Similar to the previous sections, we assume that we sample a set of instantiations $\mathcal{U}_c$ from the probability distribution $\PP$ over the parameter space $\paramspace[\dtmc]$.
In this case, we modify the LP $\mathcal{L}_r({\mathcal{U}})$ to obtain the following LP, which we denote by $\mathcal{L}_c({\mathcal{U}_c}),$ 

\begin{equation}
	\label{eq:robustmc_cost_sample}
	\begin{array}{llllllll}
	&\displaystyle\minimize \;\;\displaystyle  \;\;\tau \\
	&\subjectto\quad \forall u \in \mathcal{U}_c,\\
		& \displaystyle  c^{u}_{\sinit}\leq \tau,\\
	& \displaystyle  c^{u}_{\sinit}\leq \kappa,\\
	& \displaystyle c^{u}_s = 0\quad\forall s \in G,\\
& \displaystyle c^{u}_s = c(s)+ \sum\nolimits_{s' \in S}\mathcal{P}(s,s')[u]\;c^{u}_{s'} \quad \forall s \in S \setminus  G,
	\end{array}
\end{equation}
where for $s \in S$, $c(s) \in \RR^{\vert \mathcal{W}\vert }_{\geq 0}$ is the cost function at state $s$, $\vert \mathcal{W} \vert$ is the number of the cost parameters, and for $u \in \mathcal{U}_c $, $c^k_s$ gives the expected cost of reaching the target $G$ of the instantiated MC $\dtmc[k]$ at state $s$. 
Note that the cost parameters $\mathcal{W}$ are in the LP $\mathcal{L}_c({\mathcal{U}_c})$ as variables for the parametric cost function,.
In the scenario problem~\eqref{eq:robustmc_cost_sample}, we optimize over $c(s)$ and $c^k_s$ to minimize the maximal induced cost of the instantiated MCs.
If $c$ is an \emph{affine} function, then the optimization problem $\mathcal{L}_c({\mathcal{U}}_c)$ is convex. 
In this case, the probabilistic properties of the scenario problem are given by the following theorem.

\begin{thm}
	Let uMC $\dtmc$ and the sample set $\mathcal{U}_c \subseteq \paramspace[\dtmc]$ with $W=\vert \mathcal W\vert$, and $K = \vert \mathcal{U}_c \vert \geq  W+1$. Assume for all $u \in \mathcal{U}_c$, $\dtmc[k] \models \varphi_c$.
	\label{thm:probabilistic_bound_2}
  For a given \emph{tolerance probability} $\nu\in[0,1)$, let the associated \emph{confidence probability}
    \begin{equation}
    \label{eq:tolerance_level_2}
\alpha_\nu=\sum\nolimits_{i=0}^{W+1}\dbinom{K}{i}(1-\nu)^{K-i}\nu^i.
\end{equation}
	Then, with a probability of at least $1-\alpha_\nu$, we have 
$\satcostprobmc \geq 1-\nu$.
\end{thm}

\begin{proof}
Following the proof of Theorem 1, we define the convex set 
	\begin{equation*}
	\begin{array}{l}
			C^{\vdtmc}_{\mathcal{U}_c}(\kappa, \tau,c)=\big\lbrace(\kappa,\tau,c)\mid\quad\forall u \in \mathcal{U}_c \textup{ such that }\\
		\qquad\displaystyle c^u_{\sinit}\leq \tau,\\
			 \qquad\displaystyle  c^{u}_{\sinit}\leq \kappa,\\
	\qquad \displaystyle c^{u}_s = 0\quad\forall s \in G,\\
	 \qquad\displaystyle c^u_s =c(s)+\sum\nolimits_{s' \in S}\mathcal{P}(s,s')[u]\;c^u_{s'} \;\; \forall s \in S\setminus G\big\rbrace 
	  \end{array}
	\end{equation*}
	
	The main difference compared to the proof of Theorem 1 is that we have cost parameters in $c$ as the decision variables and we consider an expected cost specification instead of a reachability specification. Similarly to the proof of Theorem 1, we reformulate the LP $\mathcal{L}_c(\mathcal{U}_c)$ as the following convex problem
	\begin{equation}
\begin{array}{ccll}
&\minimize & \displaystyle \tau \\
&\subjectto & \displaystyle (\kappa,\tau,c)\in\RR\times\RR\times \RR^{\vert \mathcal{W} \vert},\\
&& \displaystyle (\kappa,\tau,c)\in	C^{\vdtmc}_{\mathcal{U}_c}(\kappa, \tau, c).
\end{array}
\end{equation}
This convex problem is a scenario approximation to the chance constrained problem given by
\begin{equation}
\label{eq:robustmc_true_2}
\begin{array}{ccll}
&\minimize & \displaystyle \tau \\
&\subjectto & \displaystyle (\kappa,\tau,c)\in\RR\times\RR\times\RR^{\vert \mathcal{W} \vert},\\
&& \displaystyle \PP\left((\kappa,\tau,c)\in	C^{\vdtmc}_{\paramspace[\dtmc]}(\kappa,\tau,c)\right)\geq 1-\nu.
\end{array}
\end{equation}
\noindent Therefore, similar to the Theorem 1, we obtain the desired claim.
\end{proof}

We now consider the case that we compute an instantiation of the cost variables, and some of the instantiated MCs satisfy the expected cost specification. We construct the set $\mathcal{R}_c=\mathcal{U}_c\setminus\mathcal{Q}_c$, where $\mathcal{Q}_c$ denotes the set of samples that induce MCs which violate the specification $\varphi_c$.
For this case, we obtain:

\begin{thm}
Let uMC $\dtmc$ and the sample sets $\mathcal{U}_c,\mathcal{Q}_c\subseteq \paramspace[\dtmc]$,  with  $W=\vert \mathcal W\vert$, $K = \vert \mathcal{U}_c \vert \geq  2$ and $L = \vert \mathcal{Q}\vert $. 
  For a given \emph{tolerance probability} $\nu\in[0,1)$, let the associated \emph{confidence probability}
	\label{thm:probabilistic_bound_discard_2}
    \begin{equation}
    \label{eq:tolerance_level_discard_2}
\alpha_\nu=\dbinom{l+W+1}{l}\sum\nolimits_{i=0}^{l+W+1}\dbinom{K}{i}(1-\nu)^{K-i}\nu^i. 
\end{equation}
	Then, with a probability of at least $1-\alpha_\nu$, we have $\satcostprobmc\geq 1-\nu.$
\end{thm}
\begin{proof}
The proof is similar to the proofs of Theorem~\ref{thm:probabilistic_bound_discard} and~\ref{thm:probabilistic_bound_2}, and omitted.
\end{proof}
 
\subsection{Building Scenario-Based Algorithms}

The question remains how we leverage the theoretical results to compute an estimate on the satisfaction probability to solve Problems~1 and~2.
For instance, let $\nu$  be a violation probability and $\mathcal{U}$ the sample set. 
Then, we can use Theorem~\ref{thm:probabilistic_bound_discard} or~\ref{thm:probabilistic_bound_discard_2} to compute the confidence probability $\alpha_\nu$ by using the discarding approach from~\cite{campi2011sampling}.
Similarly, for a the sample set $\mathcal{U}$ and a threshold on the confidence probability $\alpha_\nu$ we do a \emph{bisection} on $\nu$. 
Specifically, we repeatedly apply Theorem~\ref{thm:probabilistic_bound_discard} or~\ref{thm:probabilistic_bound_discard_2} for different values of $\nu \in (0, 1)$, to see if the corresponding confidence probability $\alpha_\nu$ is below the threshold. 
We then approximate the lower and upper bounds on $\nu$.

The correctness of the approach is based on scenario-based optimization. 
However, it also applies to an obtained solution by any procedure~\cite{campi2018general}.
For instance, for any obtained value for the controlled parameters, we can construct a scenario program by sampling from random parameters.
We can then apply Theorem~\ref{thm:probabilistic_bound_discard} or~\ref{thm:probabilistic_bound_discard_2} to compute the confidence probability $\alpha_\nu$ or the violation probability~$\nu$.

\paragraph{Generalization to uMDPs.} 
Recall that we want to compute the satisfaction probability for a uMDP.
The probability that for any sampled MDP we are able to synthesize a policy that satisfies the specification $\varphi_r$.
To generalize our results to uMDPs, we can modify the constraint~\eqref{eq:robustmc_spec_cons5} in the LP $\mathcal{L}_r(\mathcal{U})$ as
\begin{align}
	 \displaystyle  p^u_s \leq  \sum\nolimits_{s' \in S}\mathcal{P}(s,\act,s')[u] \cdot p^u_{s'},\quad\forall s \in S \setminus \left( T \cup \neg\exists\finally T \right),\;\;\forall \act \in \ActS(s),\label{eq:MDP_spec}
\end{align}
asserting that, for each non-target state $s \in S$ and action $\act \in \ActS(s)$, the probability induced by the \emph{minimizing policy} is an upper bound to the probability variables $p^u_s$. 
The reachability specification $\varphi_r$ is satisfied if and only if the reachability probability at the initial state induced by the minimizing policy is less than $\lambda$.
We can assert if $\varphi_r$ is satisfied by combining the constraints~\eqref{eq:MDP_spec} with the constraints~\eqref{eq:robustmc_spec_cons1}--\eqref{eq:robustmc_spec_cons4}.
Then, our theoretical results apply to the uMDPs.
\section{Numerical Examples}
We implemented the approach from Section~\ref{sec:robust} using the model checker Storm~\cite{DBLP:conf/cav/DehnertJK017} to construct and  analyze samples of MDPs.
To solve the scenario optimization problems with cost parameters, we used the SCS solver~\cite{scs}. All computations ran on a computer with 8 2.2 GHz cores, and 32 GB of RAM.

We report on a set of well-known benchmarks used in parameter synthesis~\cite{DBLP:journals/corr/abs-1903-07993} that are, for instance, available on the website of the tools \tool{PARAM}~\cite{param_sttt} or part of the \tool{PRISM} benchmark suite~\cite{KNP12b}. 
Moreover, we created a dedicated case study that is based on the aforementioned UAV example.

\subsection{Parameter Synthesis Benchmarks}
{
\setlength{\tabcolsep}{4pt}
\begin{table}[t]

\centering
\caption{Information for the benchmark instances taken from~\cite{quatmann-et-al-atva-2016}.}
\scalebox{0.76}{
\begin{tabular}{cccccrrccrrr}
	\hline
 \rule{0pt}{1.5ex}  &  &  &  & 
  & \multicolumn{2}{c}{{Model Information}}
  & \multicolumn{2}{c}{{Satisfaction Probability}}
   \\
  \cmidrule(lr){6-7} \cmidrule(lr){8-9} \cmidrule(lr){10-11}
  & benchmark
  & instance 
  & $\varphi$
  & \#pars
  & \multicolumn{1}{c}{\enskip\#states} 
  & \multicolumn{1}{c}{\#trans} 
  & \multicolumn{1}{c}{\enskip\enskip sat ($1-\nu$)} 
  & \multicolumn{1}{c}{unsat ($\nu$)}  
  \\
  \hline
\hline

  \parbox[t]{3mm}{\multirow{9}{*}{\rotatebox[origin=c]{90}{\textbf{}}}}
  & \multirow{3}{*}{brp} \rule{0pt}{1.5ex}
  &     (256,5)     & $\mathbb{P}$ & 2 &     19\,720 &      26\,627 &  \enskip\enskip\enskip 0.055   &	0.898 \\
  & &     (16,5)      & $\mathbb{E}$ & 4 &      1\,304 &       1\,731 & \enskip\enskip\enskip 0.275   &   0.676     \\
  & &     (32,5)      & $\mathbb{E}$ & 4 &      2\,600 &       3\,459 & \enskip\enskip\enskip 0.232   &    0.718     \\
\cline{2-11}
  & \multirow{2}{*}{crowds} \rule{0pt}{2.5ex}
  &     (10,5)      & $\mathbb{P}$ & 2 &    104\,512 &     246\,082 & \enskip\enskip\enskip   0.537   &  0.413         \\
  & &     (20,7)      & $\mathbb{P}$ & 2 & 45\,421\,597 & 164\,432\,797 & \enskip\enskip\enskip 0.416   &  0.534  \\
\cline{2-11}
  & \multirow{2}{*}{nand}   \rule{0pt}{2.5ex}
  &     (10,5)      & $\mathbb{P}$ & 2 &     35\,112 &       52\,647 &  \enskip\enskip\enskip 0.218   &  0.733      \\
  & &     (25,5)      & $\mathbb{P}$ & 2 &    865\,592 &   1\,347\,047    & \enskip\enskip\enskip 0.206   &   0.744     \\ 
\cline{2-11}
  \parbox[t]{3mm}{\multirow{12}{*}{\rotatebox[origin=c]{90}{\textbf{}}}}
  & \multirow{2}{*}{consensus} \rule{0pt}{2.5ex}
  &      (2,2)      & $\mathbb{P}$ & 2 &        272 &         492 & \enskip\enskip\enskip 0.280   &  0.669      \\
  & &      (4,2)      & $\mathbb{P}$ & 4 &     22\,656 &      75\,232 &  \enskip\enskip\enskip 0.063   &  0.888     \\
\hline& 
\end{tabular}
}

\label{tab:model_information}
\end{table}
}

\paragraph{Setup.} In our first set of benchmarks, we adopt parametric MDPs and MCs from~\cite{quatmann-et-al-atva-2016}. 
Essentially, the technique from that paper allows to approximate the percentage of instantiations that satisfy (or do not satisfy) a specification.
We assume a uniform distribution over the parameter space and set $\nu$ equal to the percentage of instantiations that do not satisfy the specification (and vice versa for $1-\nu$). 
We solve Problem~1 and show that the satisfaction probability is with confidence $\alpha_{\nu}$  as least as high as the approximate satisfaction percentages from~\cite{quatmann-et-al-atva-2016}.
We adapt the \emph{Consensus} protocol~\cite{consensus} and the \emph{Bounded Retransmission Protocol} (brp)~\cite{HSV94} to uMDPs; the \emph{Crowds Protocol} (crowds)~\cite{shmatikov2004probabilistic} and the \emph{NAND Multiplexing} benchmark (nand)~\cite{HJ02} become uMCs.
In Table~\ref{tab:model_information} we list the type of specification checked $(\varphi)$ and the number of parameters, states, and transitions.
We also list the satisfaction probability (as obtained in~\cite{quatmann-et-al-atva-2016}) for satisfying (sat) and falsifying (unsat) the specification $\varphi$. 

\paragraph{Results.}%
{
\setlength{\tabcolsep}{4pt}
\begin{table}[t]
\centering

\caption{Confidence probabilities $\alpha_\nu$ for different numbers of samples. 
}
\scalebox{0.760}{
\begin{tabular}{cccccccccccc}
	\hline
 \rule{0pt}{1.5ex}  & Samples  &  
  & \multicolumn{2}{c}{{100}}
  & \multicolumn{2}{c}{{1,000}}
  & \multicolumn{2}{c}{{10,000}}
  & 
   \\
 \cmidrule(lr){4-5} \cmidrule(lr){6-7} \cmidrule(lr){8-9} 
  & benchmark
  & instance 
  & \multicolumn{1}{c}{$\alpha_\nu$, sat} 
  & \multicolumn{1}{c}{$\alpha_\nu$, unsat} 
  & \multicolumn{1}{c}{$\alpha_\nu$, sat} 
  & \multicolumn{1}{c}{$\alpha_\nu$, unsat} 
  & \multicolumn{1}{c}{$\alpha_\nu$, sat} 
  & \multicolumn{1}{c}{$\alpha_\nu$, unsat}  
  & \multicolumn{1}{c}{Time (s)}  
  &
  \\
\hline
\hline

  & \multirow{3}{*}{brp}  \rule{0pt}{3.0ex}
    &     (256,5)     &    $9.99 \cdot 10^{-2}$ & $7.02 \cdot 10^{-1}$ & $1.60 \cdot 10^{-2}$ & $7.77 \cdot 10^{-2}$ & $1.12 \cdot 10^{-6}$ & $3.55 \cdot 10^{-6}$ & 1761.45 \\
  & &     (16,5)      &  $2.72 \cdot 10^{-1} $ & $ 1.97 \cdot 10^{-1}$ & $ 1.14 \cdot 10^{-1} $ & $ 3.36 \cdot 10^{-2} $ & $ 5.52 \cdot 10^{-6} $ & $ 1.80 \cdot 10^{-8} $ & 39.76 
    \\
  & &     (32,5)      &   $4.01 \cdot 10^{-1} $ & $ 2.95 \cdot 10^{-1}$ & $ 1.39 \cdot 10^{-1} $ & $ 7.76 \cdot 10^{-2} $ & $ 1.24 \cdot 10^{-6} $ & $ 2.63 \cdot 10^{-6} $ & 78.17 \\
\cline{2-11}
  & \multirow{2}{*}{crowds} 
  \rule{0pt}{3.0ex} 
  &     (10,5)      & $2.57 \cdot 10^{-1}$ & $3.72 \cdot 10^{-1}$ & $1.65 \cdot 10^{-1}$ & $1.16 \cdot 10^{-1}$ & $9.33 \cdot 10^{-7}$ & $8.22 \cdot 10^{-4}$    & 0.19  \\
  & &     (20,7)      & $4.18 \cdot 10^{-1}$ & $1.38 \cdot 10^{-1}$ & $2.41 \cdot 10^{-1}$ & $9.48 \cdot 10^{-2}$ & $5.81 \cdot 10^{-5}$ & $2.83 \cdot 10^{-5}$ & 0.45 \\
\cline{2-11}
  & \multirow{2}{*}{nand} 
    \rule{0pt}{3.0ex}   
  &     (10,5)      &   $3.48 \cdot 10^{-1}$ & $2.95 \cdot 10^{-1}$ & $3.64 \cdot 10^{-2}$ & $3.41 \cdot 10^{-1}$ & $2.64 \cdot 10^{-9}$ & $1.48 \cdot 10^{-4}$ & 144.26   \\
  & &     (25,5)      & $4.42 \cdot 10^{-1}$ & $3.71 \cdot 10^{-1}$ & $4.12 \cdot 10^{-2}$ & $3.78 \cdot 10^{-1}$ & $3.49 \cdot 10^{-6}$ & $2.91 \cdot 10^{-4}$  & 5327.82 \\ 
 \cline{2-11}
  & \multirow{2}{*}{consensus} 
    \rule{0pt}{3.0ex} 
  &      (2,2)      &   $3.38 \cdot 10^{-1}$ & $3.56 \cdot 10^{-1}$ & $1.32 \cdot 10^{-1}$ & $1.32 \cdot 10^{-1}$ & $5.67 \cdot 10^{-7}$ & $8.37 \cdot 10^{-4}$ & 0.72  \\
  & &      (4,2)      & $1.79 \cdot 10^{-1}$ & $1.41 \cdot 10^{-1}$ & $6.51 \cdot 10^{-2}$ & $4.75 \cdot 10^{-3}$ & $4.26 \cdot 10^{-5}$ & $9.29 \cdot 10^{-8}$ & 300.21    \\
\cline{2-11}
\hline& 
\end{tabular}
}

\label{table:table_confidence}

\end{table}
}
Table~\ref{table:table_confidence} shows the confidence probability $\alpha_\nu$ for each benchmark to satisfy and falsify the specification after $100, 1\,000$ and $10\,000$ samples from the parameter space.
In particular, for each number of samples, we report the average $\alpha_\nu$ after running 10 full iterations of the same benchmark.
Furthermore, we list the time to solve $1\,000$ samples for each instance (Time (s)).


The results in Table~\ref{table:table_confidence} show that for some benchmarks we get a high confidence probability already after $1\,000$ samples.
For other benchmarks, the confidence probability is still considerably low, for instance considering nand and falsifying the specification.
After $10\,000$ samples, we get a very high confidence in the satisfaction probability for all benchmarks.
These results demonstrate that we can efficiently compute a high confidence in the satisfaction probability.
In particular, for the same number of samples, the obtained confidence probabilities are consistent for varying number of states and parameters of the underlying models.
Therefore, no dependence on the size of models is shown (see Remark~\ref{remark:independent}).


\begin{figure}[t]
\centering
\scalebox{0.84}{
\input{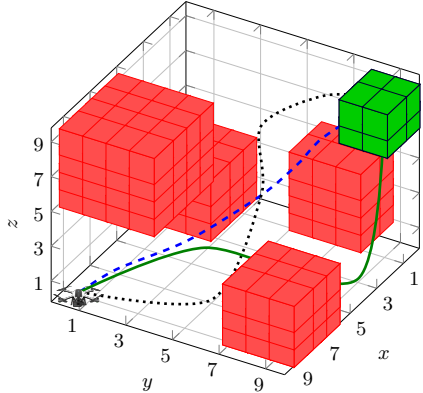}
}
\caption{An example of a 3D UAV benchmark with obstacles and a target area.}
\label{fig:drone}
\end{figure}

\subsection{UAV Motion Planning}
In our second benchmark, we consider the previously mentioned UAV motion planning example to model a realistic problem with a high number of random parameters.  
We model the problem as a uMDP, where the parameters represent how the weather conditions affect the movement of the UAV, and how the weather may change.
In particular, different wind conditions induce specific satisfaction probabilities.
We assume that the planning area is a certain valley where we have historic weather data which provide distributions over parameter values. 
The mission of the UAV is to transport a payload to a specific location and return safely to its initial position. 
The problem is to compute the satisfaction probability, that is, the probability that for any sampled MDP for this scenario we are able to synthesize a UAV policy that satisfies the specification.

We model the problem as follows:
States encode the position of the UAV, the current weather situation, and the general wind direction in the valley. 
Parameters describe how the weather affects the position of the UAV for different zones in the valley, and how the weather/wind may change during the day. 
Fig.~\ref{fig:drone} shows an example environment with zones to avoid (red) and a target zone (green).
We define four different weather conditions that each induce certain probability distributions over the eight different wind directions. 
The parameters of the model determine the probabilities of transitioning between different weather and wind conditions at each time step.
The specification is to reach the target zone safely with a probability of at least $0.9$. 
The number of states in our example is 266\,880, and the number of parameters is 2\,500.

For the distributions over parameter values, that is, over weather conditions, we consider the following cases.
First, we assume a uniform distribution over the different weather conditions in each zone. 
Second, the probability for a weather condition inducing a wind direction that pushes the UAV into the positive $y$-direction is five times more likely than others. 
Similarly, in the third case, it is five times more likely to push the UAV into the negative $x$-direction.
We depict some example trajectories of the UAV for three different conditions in Fig.~\ref{fig:drone}.
The trajectory given by the blue dashed line represents the expected trajectory for the first case, taking a  direct route to reach the target area.
Similarly, the trajectories given by the black dotted and solid green lines represent the expected trajectories for the second and third cases.
For the second case, we observe that the UAV tries to avoid to get closer to the obstacles in $x$ direction as the wind may push the UAV to the obstacles.
For the third case, the UAV avoids the obstacle at the bottom and then reaches the target area.

We sample $1\,000$ parameters for each case and approximate the maximal satisfaction probability with a confidence probability of at least $1-\alpha_\nu$, with $\alpha_\nu=10^{-6}$. 
The highest satisfaction probability is given by the first weather condition with $0.86$, and the other conditions have a satisfaction probability of $0.78$ and $0.75$, showing that it may be harder to navigate around the obstacles with non-uniform probability distributions.
The average time to compute the satisfaction probabilities is $1\,341$ seconds.

Finally, we introduce costs to a 2-dimensional example, where hitting an obstacle causes (1) a cost of $100$ and (2) the UAV to return to the initial position.
Specifically, we introduce cost parameters for transitions that steer the UAV towards $x$ or $y$-directions.
We minimize the maximal possible expected cost (under all parameter values) to reach the target location.
The specification asserts that the resulting expected cost should be less than $20$.

We uniformly sample $1\,000$ parameter values for weather conditions and note that the UAV policies favor on average transitioning to $y$-direction more compared to the $x$-direction to minimize the cost while ensuring that the probability of hitting an obstacle is minimized.
The average expected cost of the induced MDPs is $7.41$ and the satisfaction probability is $0.71$.
The solving time for this example is $2\,274$ seconds.
\section{Conclusion}
We presented a new sampling-based approach to uncertain Markov models. 
Theoretically, we showed how to effectively and efficiently approximate the probability that any randomly drawn sample satisfies a temporal logic specification. 
Furthermore, we showed the computational tractability of our approaches by means of well-known benchmarks and a new, dedicated case study. 

In the future, we plan to exploit our approaches for more involved models such as parametric extensions to continuous-time Markov chains~\cite{DBLP:journals/tse/BaierHHK03} or Markov automata~\cite{DBLP:journals/eceasst/HatefiH12}. 
Another line of future work will be a closer integration with a parameter synthesis framework.

%

%
%
%
%
%
%
%

\newpage
\bibliographystyle{plainyr}
\bibliography{literature}



\vfill

{\small\medskip\noindent{\bf Open Access} This chapter is licensed under the terms of the Creative Commons\break Attribution 4.0 International License (\url{http://creativecommons.org/licenses/by/4.0/}), which permits use, sharing, adaptation, distribution and reproduction in any medium or format, as long as you give appropriate credit to the original author(s) and the source, provide a link to the Creative Commons license and indicate if changes were made.}

{\small \spaceskip .28em plus .1em minus .1em The images or other third party material in this chapter are included in the chapter's Creative Commons license, unless indicated otherwise in a credit line to the material.~If material is not included in the chapter's Creative Commons license and your intended\break use is not permitted by statutory regulation or exceeds the permitted use, you will need to obtain permission directly from the copyright holder.}

\medskip\noindent\includegraphics{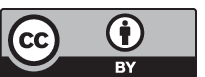}

\end{document}